\documentclass[11pt,a4paper,twocolumn]{article}
\usepackage[latin1]{inputenc}
\usepackage{amsmath}
\usepackage{amsfonts}
\usepackage{amssymb}
\usepackage{pst-all}

\newcommand\qed{\hfill\ensuremath{\Box}\smallskip}

\newtheorem{theorem}{Theorem}[section]

\newtheorem{corollary}[theorem]{Corollary}

\newtheorem{lemma}[theorem]{Lemma}

\newtheorem{construction}[theorem]{Construction}

\newenvironment{proof}{
        \noindent {\bf Proof: }}{ }

\newcommand\buchi{B\"uchi}
\newcommand\cobuchi{Co-B\"uchi}

\newcommand\A{\mathcal A}
\newcommand\B{\mathcal B}
\newcommand\C{\mathcal C}
\newcommand\D{\mathcal D}
\newcommand\E{\mathcal E}
\newcommand\F{\mathcal F}
\newcommand\G{\mathcal G}
\newcommand\scL{\mathcal L}
\newcommand\scP{\mathcal P}

\renewcommand\uplus{\dot{\cup}}

\begin{document}

\sloppy

\title{\vspace*{-4mm}Minimisation of Deterministic Parity and \buchi\ Automata and\\Relative Minimisation of Deterministic Finite Automata%
\footnote{This work was partly supported by the Engineering and Physical Science Research Council (EPSRC) through the grant EP/H046623/1 `Synthesis and Verification in Markov Game Structures'; it is an extended version of \cite{Schewe/10/minimise}.}}

\author{Sven Schewe\\University of Liverpool\\\tt sven.schewe@liverpool.ac.uk}

\date{}

\maketitle

\begin{abstract}
In this report we study the problem of minimising deterministic automata over finite and infinite words.
Deterministic finite automata are the simplest devices to recognise regular languages, and deterministic \buchi, \cobuchi, and parity automata play a similar role in the recognition of $\omega$-regular languages.
While it is well known that the minimisation of deterministic finite and weak automata is cheap, the complexity of minimising deterministic \buchi\ and parity automata has remained an open challenge.
We establish the NP-completeness of these problems.

A second contribution of this report is the introduction of almost equivalence, an equivalence class for strictly between language equivalence for deterministic \buchi\ or \cobuchi\ automata and language equivalence for deterministic finite automata.
Two finite automata are almost equivalent if they, when used as a monitor, provide a different answer only a bounded number of times in any run,
and we call the minimal such automaton relatively minimal.
Minimisation of DFAs, hyper-minimisation, relative minimisation, and the minimisation of deterministic \buchi\ (or \cobuchi) automata are operations of increasing reduction power, as the respective equivalence relations on automata become coarser from left to right.
Besides being a natural equivalence relation for finite automata, almost equivalence is language preserving for weak automata, and can therefore also be viewed as a generalisation of language equivalence for weak automata to a more general class of automata.
From the perspective of \buchi\ and \cobuchi\ automata, we gain a cheap algorithm for state-space reduction that also turns out to be beneficial for further heuristic or exhaustive state-space reductions put on top of it.
\end{abstract}

\section{Introduction}

The minimisation of deterministic finite automata (DFAs) is a classic problem with an efficient solution \cite{Hopcroft/70/minDFA,Hopcroft+Ullman}.
This report was originally written with only the question in mind of whether or not a similar result can be obtained for deterministic automata over infinite words.
Is their minimisation tractable?
For weak automata, the answer is known to be positive \cite{Loding/01/weak}, which seems to encourage a quest for a tractable solution for \buchi, \cobuchi, and parity automata as well.
However, it turns out that their minimisation is intractable (NP-complete).

This raised the question whether there are natural tractable problems between the minimisation of DFAs and deterministic \buchi\ automata (DBAs) or deterministic \cobuchi\ automata (DCAs).
The hyper-minimisation of deterministic automata \cite{BGS/09/hypermin,Badr/09/hypermin,GJ/09/hypermin,Holzer+Maletti/10/hypermin} is such an example: If we minimise a DFA while allowing for a finite symmetrical difference between the language of the source and target automaton, we might be rewarded by a smaller automaton.

We introduce a second relaxation, almost equivalence, where we require that acceptance differs only on finitely many prefixes of every infinite word.
This provides the guarantee that, on each infinite run, the result is equivalent in almost all positions (cf.\ Figure \ref{fig:minimisations}), which is not only interesting in itself, but can also be viewed as a generalisation of the minimisation problem of weak automata \cite{Loding/01/weak} to a more general class.

\begin{figure*}
\label{fig:minimisations}
\begin{center}

\psset{yunit=15mm,xunit=15mm}
\begin{pspicture}(-1.6,-.35)(1.65,3.5)

\rput(-1,0.2){\circlenode[doubleline=true]{0}{$\quad$}}

\rput(0,0){\circlenode{1}{$\quad$}}
\rput(1,0){\circlenode{2}{$\quad$}}

\rput(0,1){\circlenode{3}{$\quad$}}
\rput(1,1){\circlenode[doubleline=true]{4}{$\quad$}}

\pnode(-1.5,0.2){a}
\ncline{->}{a}{0}

\rput(0,-.55){(a)}

\ncarc[arcangle=-40,npos=.3]{->}{0}{2}
\nbput[npos=.3]{$a$}
\ncline{->}{0}{3}
\naput{$b$}

\ncarc{->}{1}{2}
\naput{$a$}
\ncarc{->}{2}{1}
\naput{$a$}

\ncarc{->}{3}{4}
\naput{$a$}
\ncarc{->}{4}{3}
\naput{$a$}

\ncline{->}{1}{3}
\naput{$b$}
\ncline{->}{2}{4}
\nbput{$b$}


\rput(0,2){\circlenode{1}{$\quad$}}
\rput(1,2){\circlenode[doubleline=true]{2}{$\quad$}}

\ncline{->}{3}{1}
\naput{$b$}
\ncline{->}{4}{2}
\nbput{$b$}

\rput(0,3){\circlenode{3}{$\quad$}}
\rput(1,3){\circlenode[doubleline=true]{4}{$\quad$}}

\ncline{->}{1}{3}
\naput{$b$}
\ncline{->}{2}{3}
\naput{$b$}

\ncarc{->}{3}{4}
\naput{$a$}
\ncarc{->}{4}{3}
\naput{$a$}
\end{pspicture}
\begin{pspicture}(-.65,-.35)(1.65,1.5)

\rput(0,0){\circlenode{1}{$\quad$}}
\rput(1,0){\circlenode{2}{$\quad$}}

\rput(0,1){\circlenode{3}{$\quad$}}
\rput(1,1){\circlenode[doubleline=true]{4}{$\quad$}}

\pnode(-0.5,0){a}
\ncline{->}{a}{1}

\rput(0.5,-.55){(b)}

\ncarc{->}{1}{2}
\naput{$a$}
\ncarc{->}{2}{1}
\naput{$a$}

\ncarc{->}{3}{4}
\naput{$a$}
\ncarc{->}{4}{3}
\naput{$a$}

\ncline{->}{1}{3}
\naput{$b$}
\ncline{->}{2}{4}
\nbput{$b$}


\rput(0,2){\circlenode{1}{$\quad$}}
\rput(1,2){\circlenode[doubleline=true]{2}{$\quad$}}

\ncline{->}{3}{1}
\naput{$b$}
\ncline{->}{4}{2}
\nbput{$b$}

\rput(0,3){\circlenode{3}{$\quad$}}
\rput(1,3){\circlenode[doubleline=true]{4}{$\quad$}}

\ncline{->}{1}{3}
\naput{$b$}
\ncline{->}{2}{3}
\naput{$b$}

\ncarc{->}{3}{4}
\naput{$a$}
\ncarc{->}{4}{3}
\naput{$a$}
\end{pspicture}
\begin{pspicture}(-.65,-.35)(1.65,1.5)

\rput(0,0){\circlenode{1}{$\quad$}}
\rput(1,0){\circlenode{2}{$\quad$}}

\rput(0,1){\circlenode{3}{$\quad$}}
\rput(1,1){\circlenode[doubleline=true]{4}{$\quad$}}

\pnode(-0.5,0){a}
\ncline{->}{a}{1}

\rput(0.5,-.55){(c)}

\ncarc{->}{1}{2}
\naput{$a$}
\ncarc{->}{2}{1}
\naput{$a$}

\ncarc{->}{3}{4}
\naput{$a$}
\ncarc{->}{4}{3}
\naput{$a$}

\ncline{->}{1}{3}
\naput{$b$}
\ncline{->}{2}{4}
\nbput{$b$}


\rput(0.5,2){\circlenode[doubleline=true]{2}{$\quad$}}

\ncline{->}{3}{2}
\naput{$b$}
\ncline{->}{4}{2}
\nbput{$b$}

\rput(0,3){\circlenode{3}{$\quad$}}
\rput(1,3){\circlenode[doubleline=true]{4}{$\quad$}}

\ncline{->}{2}{3}
\naput{$b$}

\ncarc{->}{3}{4}
\naput{$a$}
\ncarc{->}{4}{3}
\naput{$a$}
\end{pspicture}
\begin{pspicture}(-.65,-.35)(1.6,1.5)

\rput(0,0){\circlenode{1}{$\quad$}}

\rput(0,1){\circlenode[doubleline=true]{4}{$\quad$}}

\pnode(-0.5,0){a}
\ncline{->}{a}{1}

\rput(0,-.55){(d)}

\nccircle[angleA=270]{->}{1}{0.3}
\nbput{$a$}

\ncline{->}{1}{4}
\naput{$b$}

\nccircle[angleA=270]{->}{4}{0.3}
\nbput{$a$}


\rput(0,2){\circlenode{2}{$\quad$}}

\ncline{->}{4}{2}
\naput{$b$}

\rput(0,3){\circlenode[doubleline=true]{4}{$\quad$}}

\ncline{->}{2}{4}
\naput{$b$}

\nccircle[angleA=270]{->}{4}{0.3}
\nbput{$a$}
\end{pspicture}

\end{center}

\caption{%
Figure \ref{fig:minimisations}a shows a minimal DFA $\mathcal A$ over the two letter alphabet $\{a,b\}$.
Figure \ref{fig:minimisations}b shows a hyperminisation of $\mathcal A$,
Figure \ref{fig:minimisations}c shows a minimal almost equivalent automaton to $\mathcal A$, and
Figure \ref{fig:minimisations}d shows a minimal language equivalent DBA to $\mathcal A$.
(Neither hyperminimal nor relative minimal automata need to be unique.)
Hyperminimal automata are the minimal automata with a finite symmetric language difference to the source automaton \cite{BGS/09/hypermin,Badr/09/hypermin,GJ/09/hypermin}; they may only differ from the minimal automaton in the preamble.
(The trivial SCCs of the automaton reachable from the initial state.)
Minimal almost equivalent automata only guarantee that the symmetrical difference intersected with the prefixes of every infinite word are finite.
(In both cases, finite implies bounded by the number of states of the automaton.)
For weak automata---automata whose language is equivalent when read as DBA or DCA---a minimal almost equivalent automaton is also a minimal weak automaton.}
\end{figure*}
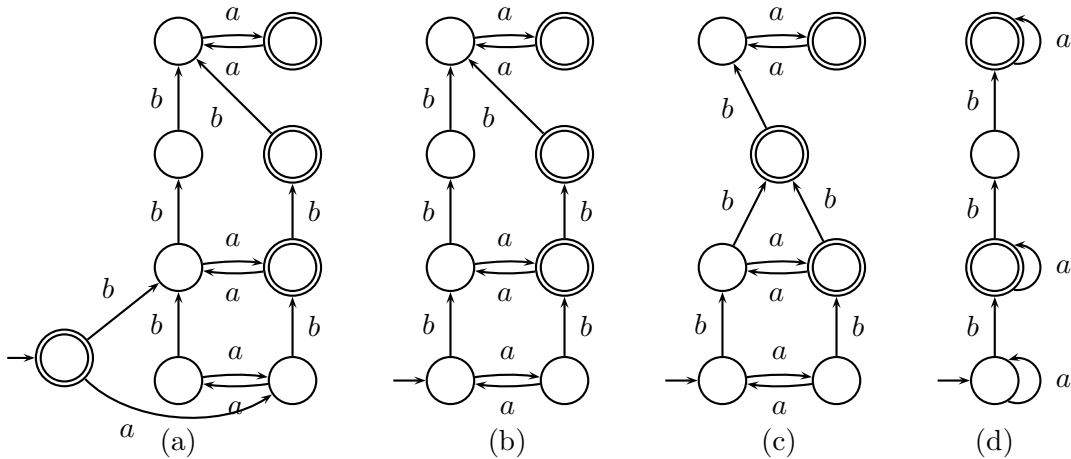

This is a natural notion of almost equivalence on DFAs, which also forms a promising basis for state-space reduction of \buchi\ and \cobuchi\ automata.
Different to the NP-completeness of minimising \buchi\ and \cobuchi\ automata, we show that finding a minimal almost equivalent DFA is cheap.
It is also a useful starting point for a state-space reduction of a DBA or DCA $\A$, because minimisation with respect to almost equivalence (like minimisation and hyper-minimisation) of $\A$ when read as a DFA are language preserving.

The algorithm we develop for finding a minimal almost equivalent DFA allows for more:
It can be strengthened by using language equivalence on $\A$ (when read as a \buchi\ or \cobuchi\ automata) in the algorithm, which provides for a smaller---yet still language equivalent---target automaton, and this automaton comes with the interesting property that one can focus on its strongly connected components (SCCs) in isolation when trying to reduce its state-space further.

While the NP-completeness of the minimisation problem of DBAs, DCAs, and deterministic parity automata (DPAs) seems to rule out the use of state-space reduction on large scale problems, this reduction technique therefore suggests that one might often get far on the way of reducing the state-space without having to pay a high price, while getting for free a division of the remaining potential parts of the automaton for further reduction.

This is fortunate, because the standard verification technique for the verification of Markov decision processes against LTL specifications \cite{Baier+Katoen/08/ModelChecking} as well as the synthesis of distributed systems from LTL specifications \cite{Rabin/72/Automata,Pnueli+Rosner/89/Synthesis,Pnueli+Rosner/89/Asynchronous,Schewe+Finkbeiner/06/Asynchronous,Piterman/07/Parity,Schewe/09/determinise}
require working with these deterministic $\omega$-automata, and techniques for the minimisation, or, indeed, for the state-space reduction of the automata involved are more than welcome.
The argument in favour of such reductions becomes even stronger for algorithms that synthesise distributed systems \cite{Pnueli+Rosner/90/Distributed,Kupferman+Vardi/01/Synthesizing,Madhusudan+Thiagarajan/01/Local,Walukiewicz+Mohalik/03/Distributed,Finkbeiner+Schewe/05/Distributed}, where deterministic automata occur in various steps of the construction.

\paragraph{\bf Organisation of the Report.}
In the following section, we introduce the basic notions of deterministic automata over finite and infinite words.
In Section \ref{sec:NPC} we establish the NP-completeness of the minimisation problems for deterministic \buchi, \cobuchi, and parity automata.
In Section \ref{sec:tractable} we introduce the problem of relative minimisation of deterministic finite automata, show that it is more powerful than hyper-minimisation and can be used to reduce the state-space of a deterministic \buchi\ or \cobuchi\ automaton, and develop an algorithm for the relative minimisation of deterministic finite automata.
We then show that this algorithm can be strengthened further when used for reducing the state-space of \buchi\ and \cobuchi\ automata, and that the structure of the resulting automaton is beneficial for further state-space reductions.

\section{Deterministic Automata}

\paragraph{$\omega$-Automata.}
Parity automata are word automata that recognise the $\omega$-regular languages over finite set of symbols.
A \emph{deterministic parity automaton} is a tuple $\scP = (\Sigma,Q,q_0,\delta,\pi)$, where
\begin{itemize}
\item $\Sigma$ denotes a finite set of symbols,
\item $Q$ denotes a finite set of states,
\item $q_0 \in Q_+$ with $Q_+ = Q \uplus \{\bot,\top\}$ denotes a designated initial state,
\item $\delta: Q_+ \times \Sigma \rightarrow Q_+$ is a function that maps pairs of states and input letters to either a new state, or to $\bot$ (false, immediate rejection, blocking) or $\top$ (true, immediate acceptance)%
\footnote{The question whether or not an automaton can immediately accept or reject is a matter of taste. Often, immediate rejection is covered by allowing $\delta$ to be partial while there is no immediate acceptance.
For technical convenience, we allow both, but treat $\top$ and $\bot$ as accepting and rejecting sink states, respectively.},
such that $\delta(\top,\sigma)=\top$ and $\delta(\bot,\sigma)=\bot$ hold for all $\sigma\in \Sigma$, and

\item $\pi: Q_+ \rightarrow P \subset \mathbb{N}$ is a priority function that maps states to natural numbers (mapping $\bot$ and $\top$ to an odd and even number, respectively), called their \emph{priority}.
(They are often referred to as colours.)
\end{itemize}

Parity automata read infinite input words $\alpha = a_0 a_1 a_2 \ldots \in \Sigma^\omega$.
(As usual, $\omega = \mathbb N_0$ denotes the non-negative integers.)
Their acceptance mechanism is defined in terms of runs:
The unique run $\rho =r_0r_1r_2\ldots\in {Q_+}^\omega$ of $\scP$ on $\alpha$ is the $\omega$-word that satisfies $r_0=q_0$ and, for all $i\in \omega$, $r_{i+1}=\delta(r_i,a_i)$.
A run is called \emph{accepting} if the highest number occurring infinitely often in the infinite sequence $\pi(r_0)\pi(r_1)\pi(r_2)\ldots$ is even, and \emph{rejecting} if it is odd.
An $\omega$-word is \emph{accepted} by $\scP$ if its run is accepting.
The set of $\omega$-words accepted by $\scP$ is called its \emph{language}, denoted $\scL(\scP)$.

We assume without loss of generality that $\max P \leq |Q|+1$.
(If a priority $p \succeq 2$ does not exist, we can reduce the priority of all states whose priority is strictly greater than $p$ by $2$ without affecting acceptance.)

Deterministic \buchi\ and \cobuchi\ automata---abbreviated DBAs and DCAs---are DPAs where the image of the priority function $\pi$ is contained in $\{1,2\}$ and $\{2,3\}$, respectively. In both cases, the automaton is often denoted $\A=(\Sigma,Q,q_0,\delta,F)$, where $F \subseteq Q_+$ denotes those states with priority $2$. The states in $F$ are also called \emph{final} or \emph{accepting states}, while the remaining states $Q_+\smallsetminus F$ are called \emph{rejecting} states.

\paragraph{Finite Automata.}
Finite automata are word automata that recognise the regular languages over finite set of symbols.
A \emph{deterministic finite automaton} (DFA) is a tuple $\F = (\Sigma,Q,q_0,\delta,F)$, where $\Sigma$, $Q$, $q_0$, and $\delta$ are defined a for DPAs, and $F \subseteq Q \uplus \{\top\}$ is a set of \emph{final states} that contains $\top$ (but not $\bot$).

Finite automata read finite input words $\alpha = a_0 a_1 a_2 \ldots a_n \in \Sigma^*$.
Their acceptance mechanism is again defined in terms of runs:
The unique run $\rho =r_0r_1r_2\ldots r_{n+1} \in {Q_+}^+$ of $\F$ on $\alpha$ is the word that satisfies $r_0=q_0$ and, for all $i\leq n$, $r_{i+1}=\delta(r_i,a_i)$.
A run is called \emph{accepting} if it ends in a final state (and \emph{rejecting} otherwise), a word is \emph{accepted} by $\F$ if its run is accepting, and the set of words accepted by $\F$ is called its \emph{language}, denoted $\scL(\F)$.

\paragraph{Automata Transformations \& Conventions.}
For a deterministic automaton $\A=(\Sigma,Q,q_0,\delta,F)$ or $\A=(\Sigma,Q,q_0,\delta,\pi)$ and a state $q \in Q_+$, we denote with $\A_q=(\Sigma,Q,q,\delta,F)$ or $\A_q=(\Sigma,Q,q,\delta,\pi)$, respectively, the automaton resulting from $\A$ by changing the initial state to $q$.
We also read finite automata at times as \buchi\ (or \cobuchi) automata and \buchi\ (or \cobuchi) automata as finite automata in the constructions, and let DFAs run on infinite words where this is convenient and its meaning is clear in the context.

Automata define a directed graph whose unravelling from the initial state defines the possible runs.
For an automaton $\A=(\Sigma,Q,q_0,\delta,F)$ or $\A=(\Sigma,Q,q_0,\delta,\pi)$, this is the directed graph $(Q_+,T)$ with $T=\{(p,q) \in Q_+ \times Q_+ \mid \exists \sigma \in \Sigma.\ \delta(p,\sigma)=q\}$.
When referring to the reachable states (which always means reachable from the initial state) and SCCs of an automaton, this refers to this graph.

\paragraph{\bf Emptiness and Equivalence.}

A DPA is called \emph{empty} if its language is empty and \emph{universal} if it accepts every word $\alpha\in \Sigma^\omega$.
For two automata $\scP^1 = (\Sigma,Q_1,q_0^1,\delta_1,\pi_1)$ and $\scP^2 = (\Sigma,Q_2,q_0^2,\delta_2,\pi_2)$, two states $q_1 \in Q_1$ and $q_2 \in Q_2$ are called \emph{equivalent} if $\scL(\scP^1_{q_1})=\scL(\scP^2_{q_2})$. (Equivalence of states naturally extends to the same automaton, as $\scP^1$ and $\scP^2$ are not necessarily different.)
Two automata are equivalent if their initial states are equivalent. (Or, likewise, if they recognise the same language.)

Emptiness, universality, and equivalence of parity, \buchi, and \cobuchi\ automata is computationally easy:

\begin{theorem}
\label{theo:nlc}
Language non-inclusion of two parity automata $\scP^1 = (\Sigma,Q_1,q_0^1,\delta_1,\pi_1)$ and $\scP^2 = (\Sigma,Q_2,q_0^2,\delta_2,\pi_2)$ can be checked in non-deterministic logarithmic space.
\end{theorem}

\begin{proof}
We describe how to check non-emptiness of $\scL(\scP^1)\smallsetminus\scL(\scP^2)$. ($\scL(\scP^2)\smallsetminus\scL(\scP^1)$ can be checked accordingly, and a non-deterministic machine can guess which to check.)

$\scL(\scP^1)\smallsetminus\scL(\scP^2)$ is non-empty if there is a word $\alpha = a_0 a_1 a_2 \ldots \in \Sigma^\omega$ such that the run $r_0^1 r_1^1 r_2^1\ldots$ of $\scP^1$ on $\alpha$ is accepting, while the run and $r_0^2 r_1^2 r_2^2\ldots$ of $\scP^2$ on $\alpha$ is rejecting.

A necessary condition for this is that there are positions $i \leq j$ such that $r_i^1=r_{j+1}^1$, $r_i^2=r_{j+1}^2$, and the highest priority in $r_i^1 r_{i+1}^1 \ldots r_j^1$ is even, while the highest priority in $r_i^2 r_{i+1}^2 \ldots r_j^2$ is odd.
However, the existence of two runs with this property is also a sufficient condition for $\scL(\scP^1) \nsubseteq \scL(\scP^2)$, because the word $\alpha'=a_0 a_1 \ldots a_{i-1} (a_i a_{1+1} \ldots a_j)^\omega$ is accepted by $\scP^1$ and rejected by $\scP^2$.

Consequently, we can use a non-deterministic machine that guesses $\alpha'$ on the fly, and guesses when $i$ and $j$ are reached.
All that this machine needs to store is
\begin{itemize}
\item the current state of $\scP^1$ and $\scP^2$,
\item the current guessed input letter,
\item once $i$ is reached (guessed): $r_i^1$ and $r_i^2$, and
\item the highest priority seen since it guessed being in position $i$ for both runs.

(Initialising the values to $\pi(r_i^1)$ and $\pi(r_i^2)$, and computing and storing the maximum of the current value and the priority of the respective current state.)
\end{itemize}
Upon reaching $j$ (guessed), the machine checks if the highest priority stored is even for the run of $\scP_1$, and odd for $\scP_2$.

The overall memory required is logarithmic in the size of the automaton.
\qed
\end{proof}

Reachability in a directed graph can obviously be reduced in deterministic logspace to checking language non-emptiness of a \buchi\ or \cobuchi\ automaton with only rejecting states, or non-universality of a \buchi\ or \cobuchi\ automaton with only accepting states, respectively.
Likewise, testing universality or emptiness can be reduced in deterministic logspace to checking language inclusion (in the respective direction) with a trivial \buchi\ or \cobuchi\ automaton that immediately changes to $\top$ or $\bot$, respectively, for every input letter.
With the fact that NL is closed under complementation \cite{Immerman/88/NCoNL}, this immediately implies:

\begin{corollary}
\label{cor:NLC}
Language inclusion, equivalence, emptiness, and universality of parity, \buchi, and \cobuchi\ automata and their co-problems are NL-complete.
\qed
\end{corollary}

\section{Minimising \buchi\ and Parity Automata is NP-Complete}
\label{sec:NPC}
In this section we show that the minimisation of deterministic \buchi, \cobuchi, and parity automata are NP-complete problems.
This is in contrast to the tractable minimisation of finite~\cite{Hopcroft/70/minDFA} and weak automata~\cite{Loding/01/weak}.

The hardest part of the NP-completeness proof is a reduction from the problem of finding a minimal vertex cover of a graph to the minimisation of deterministic \buchi\ automata.
For this reduction, we first define the characteristic language of a simple connected graph.
For technical convenience we assume that this graph has a distinguished initial vertex.

We show that the states of a deterministic \buchi\ automaton that recognises this characteristic language must satisfy side-constraints, which imply that it has at least $2n+k$ states, where $n$ is the number of vertices of the graph, and $k$ is the size of its minimal vertex cover.
We then show that, given a vertex cover of size $k$, it is simple to construct a deterministic \buchi\ automaton of size $2n+k$ that recognises the characteristic language of this graph. (It can be constructed in linear time and logarithmic space.)
Furthermore, we show that minimising the automaton defined by the trivial vertex cover can be used to determine a minimal vertex cover for this graph, which concludes the reduction.

We call a non-trivial ($|V|>1$) simple connected graph $\G_{v_0}=(V,E)$ with a distinguished initial vertex $v_0\in V$ \emph{nice}.
As a warm-up, we have to show that the restriction to nice graphs leaves the problem of finding a minimal vertex cover NP-complete.

\begin{lemma}
\label{lem:NPC}
The problem of checking whether a nice graph  $\G_{v_0}$ has a vertex cover of size $k$ is NP-complete.
\end{lemma}

\begin{proof}
As a special case of the vertex cover problem, it is in NP, and the problem of finding a vertex cover of size $k$ for a graph $(V,E)$ can be reduced to the problem of checking if the nice graph $\G_v=(V \uplus \{v,v'\}, E \cup \big\{\{w,v\}\mid w \in V \uplus \{v'\}\big\}$ has a vertex cover of size $k+1$:
A vertex cover of $\G_v$ must contain a vertex cover of $(V,E)$ and $v$ or $v'$, and a vertex cover of $(V,E)$ plus $v$ is a vertex cover \linebreak of $\G_v$.
\qed
\end{proof}

We define the \emph{characteristic language} $\scL(\G_{v_0})$ of a nice graph $\G_{v_0}$ as the $\omega$-language over $V_\natural=V\uplus\{\natural\}$ (where $\natural$ indicates a stop of the evaluation in the next step---it can be read `stop') consisting of
\begin{enumerate}
\item all $\omega$-words of the form ${v_0}^*{v_1}^+ {v_2}^+ {v_3}^+ {v_4}^+ \ldots \in V^\omega$ with $\{v_{i-1},v_i\}\in E$ for all $i \in \mathbb N$, (words where $v_0,v_1,v_2,\ldots$ form an infinite path in $\G_{v_0}$), and 

\item all $\omega$-words starting with ${v_0}^*{v_1}^+ {v_2}^+ \ldots {v_n}^+ \natural v_n \in {V_\natural}^*$ with $n \in \mathbb N_0$ and $\{v_{i-1},v_i\}\in E$ for all $i \in \mathbb N$.
(Words where $v_0,v_1,v_2,\ldots,v_n$ form a finite---and potentially trivial---path in $\G_{v_0}$, followed by a $\natural$ sign, followed by the last vertex of the path $v_0,v_1,v_2,\ldots,v_n$.)
\end{enumerate}
We call the $\omega$-words in (1) \emph{trace-words}, and those in (2) \emph{$\natural$-words}. The trace-words are in $V^\omega$, while the $\natural$-words are in ${V_\natural}^\omega \smallsetminus V^\omega$.

Let $\B$ be a deterministic \buchi\ automaton that recognises the characteristic language of $\G_{v_0}=(V,E)$.
We call a state of $\B$
\begin{itemize}
\item a \emph{$v$-state} if it can be reached upon an input word ${v_0}^*{v_1}^+ {v_2}^+ \ldots {v_n}^+ \in {V_\natural}^*$, with $n \in \mathbb N_0$ and $\{v_{i-1},v_i\}\in E$ for all $i \in \mathbb N$, that ends in $v=v_n$ (in particular, the initial state of $\B$ is a $v_0$-state), and
\item a \emph{$v\natural$-state} if it can be reached from a $v$-state upon reading a $\natural$ sign.
\end{itemize}
We call the union over all $v$-states the set of \emph{vertex-states}, and the union over all $v\natural$-states the set of $\natural$-states.

\begin{lemma}
Let $\mathcal G_{v_0}=(V,E)$ be a nice graph with initial vertex $v_0$, and let $\B=(V,Q,q_0,\delta,F)$ be a deterministic \buchi\ automaton that recognises the characteristic language of $\G_{v_0}$.
Then (1) the vertex- and $\natural$-states of $\B$ are disjoint, and, for all $v,w \in V$ with $v\neq w$, (2) the $v$-states and $w$-states and (3) the $v\natural$- and $w\natural$-states are disjoint.
For each vertex $v\in V$, there is (4) a $v\natural$-state and (5) a rejecting $v$-state, and (6), for every edge $\{v,w\}\in E$, there is an accepting $v$-state or an accepting $w$-state.
\end{lemma}

\begin{proof}
\begin{enumerate}
\item Let $q_v^\natural$ be a $v\natural$-state and $q$ a vertex-state.
As $\B$ recognises $\scL( \G_{v_0})$, $\B_{q_v^\natural}$ must accept $v^\omega$, while $\B_{q}$ must reject it.

\item Let $q_v$ be a $v$-state and let $q_w$ be a $w$-state with $v\neq w$.
As $\B$ recognises $\scL( \G_{v_0})$, $\B_{q_v}$ must accept $\natural v^\omega$, while $\B_{q_w}$ must reject it.

\item Let $q_v^\natural$ be a $v\natural$-state and let $q_w^\natural$ be a $w\natural$-state with $v\neq w$.
As $\B$ recognises $\scL( \G_{v_0})$, $\B_{q_v^\natural}$ must accept $v^\omega$, while $\B_{q_w^\natural}$ must reject it.

\item As $\G_{v_0}$ is connected, there is, for every $v\in V$, a path $v_0v_1v_2\ldots v$ in $\G_{v_0}$, and the state reached by $\B$ upon reading $v_1v_2\ldots v \natural$ is a $v\natural$-state.

\item As $\G_{v_0}$ is connected, there is, for every $v\in V$, a path $v_0v_1v_2\ldots v$ in $\G_{v_0}$.
After reading $v_1v_2\ldots v$, $\B$ is in a $v$-state.
$\B$ remains in $v$-states if it henceforth reads $v$'s.
(Note that the automaton cannot block/reject immediately, as it should accept a continuation $\natural v^\omega$ at any time.)
As the word is rejecting, almost all states in the run of the automaton are rejecting $v$-states.

\item Let us consider an arbitrary edge $\{v,w\}$.
As $\G_{v_0}$ is connected, there is a path from $v_0v_1v_2\ldots v$ in $\G_{v_0}$, and $v_1v_2\ldots v (wv)^\omega$ is in $\scL( \G_{v_0})$;
the run of $\B$ on this $\omega$-word is therefore accepting.
As almost all states in this accepting run are $v$-states or $w$-states, there must be an accepting $v$-state or an accepting $w$-state.
\qed
\end{enumerate}
\end{proof}

The sixth claim implies that the set $C$ of vertices with an accepting vertex-state is a vertex cover of $\mathcal G_{v_0}=(V,E)$.
It is also clear that $\B$ has at least $|V|$ rejecting vertex-states, $|C|$ accepting vertex-states, and $|V|$ $\natural$-states:

\begin{corollary}
\label{cor:atLeast}
For a deterministic \buchi\ automaton that recognises the characteristic language of a nice graph $\mathcal G_{v_0}=(V,E)$ with initial vertex $v_0$, the set $C=\{v\in V \mid$ there is an accepting $v$-state$\}$ is a vertex cover of $\mathcal G_{v_0}$, and $\B$ has at least $2|V| + |C|$ \linebreak states.
\qed
\end{corollary}

It is not hard to define, for a given nice graph $\mathcal G_{v_0}=(V,E)$ with vertex cover $C$, a \buchi\ automaton $\B^{\G_{v_0}}_C= (V_\natural,(V\times\{r,\natural\}) \uplus (C \times \{a\}),(v_0,r),\delta,(C\times\{a\})\uplus \{\top\})$ with $2|V|+|C|$ states that recognises the characteristic language of $\mathcal G_{v_0}$:
We simply choose
\begin{itemize}
\item $\delta\big((v,r),v'\big) = (v',a)$ if $\{v,v'\} \in E$ and $v'\in C$,

$\delta\big((v,r),v'\big) = (v',r)$ if $\{v,v'\} \in E$ and $v'\notin C$,

$\delta\big((v,r),v'\big) = (v,r)$ if $v=v'$,

$\delta\big((v,r),v'\big) = (v,\natural)$ if $v'=\natural$, and

$\delta\big((v,r),v'\big)=\bot$ otherwise;
\item $\delta\big((v,a),v'\big) = \delta\big((v,r),v'\big)$, and
\item $\delta\big((v,\natural),v\big) = \top$ and  $\delta\big((v,\natural),v'\big) = \bot$ for $v' \neq v$.
\end{itemize}

$\B^{\G_{v_0}}_C$ simply has one $v\natural$-state for each vertex $v\in V$ of $\mathcal G_{v_0}$, one accepting $v$-state for each vertex in the vertex cover $C$, and one rejecting $v$-vertex for each vertex $v\in V$ of $\mathcal G_{v_0}$.
It moves to the accepting copy of a vertex state $v$ only upon taking an edge to $v$, but not on a repetition of~$v$.

\begin{lemma}
\label{lem:correct}
For a nice graph $\G_{v_0}=(V,E)$ with initial vertex $v_0$ and vertex cover $C$, $\B^{\G_{v_0}}_C$ recognises the characteristic language of $\mathcal G_{v_0}$.
\end{lemma}

\begin{proof}
To show $\scL(\B^{\G_{v_0}}_C) \subseteq \scL(\G_{v_0})$, let us consider an $\omega$-word $\alpha$ accepted by $\B^{\G_{v_0}}_C$.
Then it is either eventually accepted immediately when reading a $v$ from a state $(v,\natural)$, or by seeing accepting states in $C\times \{a\}$ infinitely many times.
By the construction of $\B^{\G_{v_0}}_C$, $\alpha$ must be a $v\natural$-word in the first case, and a trace-word in the latter.

To show $\scL(\B^{\G_{v_0}}_C) \supseteq \scL(\G_{v_0})$, it is apparent that $\natural$-words are accepted immediately after reading the initial sequence that makes them $\natural$-words, while a trace-word
${v_0}^{i_0-1} {v_1}^{i_1} {v_2}^{i_2} {v_3}^{i_3} \ldots \in V^\omega$ with $i_j \in \mathbb N$ and $\{v_j,v_{j+1}\}\in E$ for all $j \in \omega$, has the run
$\rho=(v_0,r)^{i_0} (v_1,p_1) (v_1,r)^{i_1-1} (v_2,p_2) (v_2,r)^{i_2-1} (v_3,p_3) \ldots$, with $p_i=a$ (and hence $(v_i,p_i)$ accepting) if $v_i$ in $C$.
As $C$ is a vertex cover, this is at least the case for every second index. (There is no $n\in \mathbb N$ with $\{v_n,v_{n+1}\} \cap C = \emptyset$.)
$\rho$ therefore contains infinitely many accepting states.
\qed
\end{proof}

Corollary \ref{cor:atLeast} and Lemma \ref{lem:correct} immediately imply:

\begin{corollary}
\label{cor:reduce}
Let $C$ be a minimal vertex cover of a nice graph $\mathcal G_{v_0}=(V,E)$. Then $\B^{\G_{v_0}}_C$ is a minimal deterministic \buchi\ automaton that recognises the characteristic language of $\mathcal G_{v_0}$.
\qed
\end{corollary}

From here, it is a small step to the main theorem of this section:

\begin{theorem}
The problem of whether there is, for a given deterministic \buchi\ automaton, a language equivalent \buchi\ automaton with at most $n$ states is NP-complete.
\end{theorem}

\begin{proof}
For containment in NP, we can simply use non-determinism to guess such an automaton. Checking that it is language equivalent is then in NL by Corollary \ref{cor:NLC}.

By Corollary \ref{cor:reduce}, we can reduce checking if a nice graph $G_v$ with $m$ vertices has a vertex cover of size $k$ to checking if the deterministic \buchi\ automaton $\B^{G_v}_V$---which has $3m$ states and is easy to construct (in deterministic logspace)---has a language equivalent \buchi\ automaton with $2m+k$ states.
As the problem we reduced from is NP-complete by Lemma \ref{lem:NPC}, this concludes the reduction.
\qed
\end{proof}

As minimising \cobuchi\ automata coincides with minimising the dual \buchi\ automata, the similar claim holds for \cobuchi\ automata.

\begin{corollary}
The problem of whether there is, for a given deterministic \cobuchi\ automaton $\C$, a language equivalent \cobuchi\ automaton with at most $n$ states is NP-complete.
\qed
\end{corollary}

The problem of minimising deterministic parity automata cannot be easier than the problem of minimising \buchi\ automata, and the `in NP' argument that we can simply guess a language equivalent DPA and then inexpensively check correctness (by Corollary \ref{cor:NLC}) extends to parity automata.

\begin{corollary}
The problem of whether there is, for a given parity automaton, a language equivalent parity automaton with $n$ states is NP-complete.
\qed
\end{corollary}

Note that, while there is a minimal number of priorities required for every language, the number of states cannot be reduced by increasing the number of priorities, and minimising the number of priorities can be done in polynomial time, changing only the priority functions \cite{Niwinski+Walukiewicz/98/Hierarchies,Carton+Maceiras/99/ParityIndex}.

\section{Relative DFA Minimisation}
\label{sec:tractable}

Minimisation techniques for deterministic finite automata can be used to minimise deterministic \buchi\ and \cobuchi\ automata.
They are cheap---Hopcroft's algorithm works in time $\mathcal O(n \log n)$~\cite{Hopcroft/70/minDFA}---and have proven to be powerful devices for state-space reduction.
From a practical point of view, this invites---in the light of the intractability result for minimising deterministic \buchi\ and \cobuchi\ automata---the question if such tractable minimisation techniques can be used for a space reduction of \buchi\ and \cobuchi\ automata.
From a theoretical point of view, this invites the question of whether there are interesting tractable minimisation problems between the minimisation (or hyper-minimisation \cite{BGS/09/hypermin,Badr/09/hypermin,GJ/09/hypermin,Holzer+Maletti/10/hypermin}) of finite automata, and the minimisation of \buchi\ and \cobuchi\ automata.

Both the theoretical and the practical question turn out to have a positive answer:
An answer to the theoretical question is that we can define \emph{almost equivalence} on automata and their states as a relation, where two automata or states are almost equivalent if their language intersected with the initial sequences of every omega word have finite difference.
We show that a minimal almost equivalent automaton is easy to construct.
Besides being interesting on their own account (for example, if we want to construct a monitor that errs only a bounded number of times for every input word), they are language preserving for deterministic \buchi\ and \cobuchi\ automata.
What is more, a minimal almost equivalent automata to a weak automaton (an automaton that recognises the same language as DBA and DCA) is a minimal language equivalent weak automaton.

From a practical point of view, the algorithm suggests an approximation that is valid for both \buchi\ and \cobuchi\ automata.
There is, however, a simple and apparent improvement of the algorithm when used for the minimisation of \buchi\ and \cobuchi\ automata:
Instead of almost equivalence of states, we can use language equivalence for \buchi\ or \cobuchi\ automata, respectively.
But the algorithm provides for more:
It isolates the minimisation problem within in the SCC.
That is, both precise and approximative minimisation techniques can look into these simpler sub-structures.

While being language preserving when the DBA or DCA is read as a DFA is a sufficient criterion for language preservation of the automaton itself, it is by no means necessary.
In this context it becomes apparent that the NP-completeness result of the previous section may not hint at the fact that state-space reduction for DBAs and DCAs is beyond price; one should rather take it as a hint that a high price might have to be paid for the \emph{additional} benefit one can get from stronger state-space reductions than those for DFAs.

However, even if we consider DFAs, there is at time a desire for stronger reductions than language preserving minimisation.
For this reason, hyper-minimisation, the problem of finding a minimal automaton with a finite symmetrical difference in its language, has been studied for DFAs \cite{BGS/09/hypermin,Badr/09/hypermin,GJ/09/hypermin,Holzer+Maletti/10/hypermin}.
In this section, we introduce relative minimisation where we seek a minimal automaton for which the symmetrical difference intersected with the initial sequences of every infinite word is bounded.
The underlying notion of approximate equivalence is weaker than the $f$-equivalence used for hyper-minimisation, and in my opinion it is also more natural even for DFAs.
(One is often not really interested in differences on words that one never observes.)
It surely is the better starting point for minimising DBAs and DCAs.
We develop a simple algorithm for relative minimisation, and discuss how it can be strengthened to approximate minimal DBAs or DCAs even better.

\paragraph{Almost Equivalence. }
For two (not necessarily different) DFAs $\A^1 = (\Sigma,Q_1,q_0^1,\delta_1,F_1)$ and $\A^2 = (\Sigma,Q_2,q_0^2,\delta_2,F_2)$, we call two states $q_1 \in Q_1$ and $q_2 \in Q_2$ \emph{almost equivalent} if, for all $\omega$-words $\alpha \in \Sigma^\omega$,
it holds that for the runs $r_0^1 r_1^1 r_2^1 r_3^1 \ldots$ and $r_0^2 r_1^2 r_2^2 r_3^2 \ldots$ of $\A^1_{q_1}$ and $\A^2_{q_2}$ on $\alpha$, membership of the states in the final states is equivalent almost everywhere ($\exists n\in \omega.\ \forall i \geq n.\ r_i^1 \in F_1 \Leftrightarrow r_i^2\in F_2$).
Two DFAs are called almost equivalent if their initial states are, and we extend these definitions to DBAs and DCAs.

Obviously, almost equivalence is a congruence and hence defines quotient classes on the states of automata. It is also easy to compute:

\begin{lemma}
Testing almost equivalence (or inequivalence) of two DFAs $\A$ and $\B$ is NL-complete, and the quotient class of a DFA $\A$ can be computed in time quadratic in the size of the automaton.
\end{lemma}

\begin{proof}
It is simple to construct in deterministic logspace an automaton $\A\intercal\B$ whose states are ordered pairs of $\A$ and $\B$ states, with the pair of initial states of $\A$ and $\B$ as initial state, whose final states are the pairs of a final and a non-final state (where the final state might be an $\A$ or a $\B$ state).
Two states $q_a$ and $q_b$ are obviously almost equivalent if, and only if, the language of $(\A \intercal \B)_{(q_a,q_b)}$ is empty when read as a DBA, which is in NL by Corollary~\ref{cor:NLC}.
For completeness, it is again easy to reduce the reachability problem of directed graphs to refuting almost equivalence of two automata.

This simple construction also caters for a quadratic deterministic algorithm for finding the quotients of almost equivalent states:
We can construct $\A\intercal\A$ in quadratic time and find the SCCs in $\A\intercal\A$ in time linear in $\A\intercal\A$.
Two states $p,q$ are obviously either almost equivalent or one can reach a final state in a non-trivial SCC from $(p,q)$ in $\A\intercal\A$, and these states can be computed in time linear in $\A\intercal\A$ by a simple fixed-point algorithm.
\qed
\end{proof}

\paragraph{Finding minimal almost equivalent automata is tractable.}
We call the problem of finding a minimal automaton almost equivalent to a DFA $\A$ \emph{relative minimisation}.
Besides the usefulness of relative minimisation for DFAs themselves, let us consider the usefulness of relative minimisation for the state-space reduction of deterministic \buchi\ and \cobuchi\ automata.

\begin{lemma}
\label{lem:coarser}
Two deterministic \buchi\ and \cobuchi\ automata that are, when read as deterministic finite automata, almost language equivalent recognise the same language.
\end{lemma}

\begin{proof}
The priority of the states in their runs differs in only finitely many positions.
\qed
\end{proof}

We can therefore use the inexpensive DFA minimisation, hyper-minimisation (which in particular results in an almost equivalent automaton), and the newly introduced relative minimisation of DFAs for a state-space reduction of DBAs and DCAs.
This provides the back-bone for efficient relative minimisation:
To find, for a given DFA $\A = (\Sigma,Q',q_0',\delta'',F')$, a minimal deterministic automaton $\D$ that accepts an almost equivalent language, we execute the following algorithm:

\begin{construction}
\label{algorithm}
In a first step\footnote{This step is not necessary for the correctness of the algorithm or for its complexity.},
we construct the minimal language equivalent automaton $\B= (\Sigma,Q,q_0,\delta,F)$ in quasi-linear time using Hopcroft's algorithm \cite{Hopcroft/70/minDFA}.

For $\B$, we then introduce a pre-order $(Q_+,\succeq)$ on the states of $\B$ such that (1) two states are equivalent if, and only if, they are in the same SCC of $\B$, (2) if $p$ is reachable from $q$ then $p\succeq q$, and (3) $\top$ and $\bot$ are bigger than all states in $Q$.
(This can obviously be done in linear time.)

In a third step, we determine the quotient classes of almost equivalent states of $\B$, and pick, for each quotient class $[q]$, a representative $r_{[q]}\in[q]$ that is maximal with respect to $\succeq$ among the states almost equivalent to $q$.

We then construct an automaton $\C= (\Sigma,Q,r_{[q_0]},\delta',F)$ by choosing the representative $r_{[q_0]}$ of the quotient $[q_0]$ of states almost equivalent to the initial state as new initial state, and changing all transitions that lead to states whose representative is bigger (with respect to $\succeq$) to the representatives of these states.
That is, for $\delta(q,\sigma)=q'$, we get $\delta'(q,\sigma)=q'$ if $q \simeq r_{[q']}$ and $\delta'(q,\sigma)= r_{[q']}$ otherwise.

Finally, we minimise $\C$ using Hopcroft's algorithm again, yielding a DFA $\D$.
\end{construction}

\begin{lemma}
\label{lem:almost}
The DFAs $\A$ and $\D$ of the above construction are almost equivalent.
\end{lemma}

\begin{proof}
First, $\A$ and $\B$ are language equivalent.

To compare the language of $\B$ and $\C$, we note that, if $p$ and $q$ are almost equivalent, then so are $\delta(p,\sigma)$ and $\delta(q,\sigma)$ for all $\sigma$ in $\Sigma$. (Assuming the opposite, there would be a word $\alpha\in \Sigma^\omega$ for which priority of the runs of $\B_{\delta(p,\sigma)}$ and $\B_{\delta(q,\sigma)}$ differ on infinitely many positions, which implied the same for $\sigma\cdot\alpha$ and runs on $\B_p$ and $\B_q$ and hence lead to a contradiction.)

Let us now consider runs $r_0^b r_1^b r_2^b r_3^b \ldots$ and $r_0^c r_1^c r_2^c r_3^c \ldots$ of $\B$ and $\C$ on some $\omega$-word $\alpha$.
Then $r_i^b$ and $r_i^c$ are almost equivalent for all $i \in \omega$ by the above observation.
Also, states in a run of $\C$ can never go down in the pre-order $(Q_+,\succeq)$. In particular, there is a bounded (at most $|Q|$) number of positions in the run, where $\C$ takes an adjusted transition---a transition $\delta'(q,\sigma) \neq \delta(q,\sigma)$---as this involves going strictly up in $(Q_+,\succeq)$.
The number of positions $i\in \omega$ where either only $r_i^b$ or only $r_i^c$ are final can thus be estimated by the number of changed transitions taken times the bounded number of differences that can occur between almost equivalent states in $\B$.

Finally, $\C$ and $\D$ are again language equivalent.
\qed
\end{proof}

An key observation for the proof that $\D$ is minimal is that almost equivalent states are in the same quotient class.

\begin{lemma}
\label{lem:SCC}
Two states of $\D$ that are almost equivalent are in the same SCC.
\end{lemma}

\begin{proof}
As $\D$ is but the smallest automaton language equivalent to $\C$, we can obtain $\D$ by first constructing a language equivalent DFA $\C'$ from $\C$ by deleting the unreachable states of $\C$ and then joining the language equivalent states.

In $\C$, all states that have an almost equivalent peer in a bigger (by $\succeq$) SCC are unreachable:
In the construction of $\C$, their incoming transitions have been re-routed to the representative of their class, and the initial state has been swapped to the representative.
In $\C'$, all almost equivalent states are therefore in the same SCC.
The same holds true for language equivalent states, as language equivalence is the finer relation.
If two states are connected in $\C'$, the same holds for their quotients of language equivalent states \linebreak in $\D$.
\qed
\end{proof}

The proof that $\D$ is minimal builds on the fact that, whenever we go up in $(Q_+,\succeq)$, we choose the same representative.

\begin{theorem}
\label{theo:rmin}
There is no DFA $\E$ almost equivalent to $\D$ that is strictly smaller than $\D$.
\end{theorem}

\begin{proof}
For convenience, we now look at quotient classes of almost equivalent states that cover both $\D$ and $\E$ in this proof.

First, as $\D$ is minimal (among the language equivalent automata), all states in $\D$ are reachable.
Let us assume that there is a smaller DFA $\E$ almost equivalent to $\D$.
Then $\E$ must (at least) have the same quotient classes as $\D$, and hence, there must be a particular quotient class $[q]$ of $\D$ (and $\E$), such that there are strictly less representatives of this class in $\E$ than in $\D$.

By the previous lemma, the representatives of quotient classes of almost equivalent states of $\D$ are all in the same SCC.
For trivial SCCs, this implies that there is only one representative in $\D$ and hence at least as many in $\E$.

For non-trivial SCCs, there is a witness of language non-equivalence that does not leave the SCC for all different occurrences.
(Note that Construction \ref{algorithm} guarantees for $\C$ that, once an SCC is left, the target state---and hence the remainder of the run---is the same, no matter from which representative of a quotient class we start. And the proof of the previous lemma showed that the minimisation of $\C'$ is SCC preserving.)

As $\E$ has less representatives, we can pick one representative $r\in [q]$ of this class in $\D$ such that, for all representatives $e \in [q]$ in $\E$, we construct a finite word $\alpha_e \in \Sigma^*$ that is accepted either only by $\E_e$ or only by $\D_r$, such that the run of $\D_r$ on $\alpha_e$ stays in the SCC containing $r$.
This invites a simple pumping argument:
We can construct a word starting with a sequence $\beta_0$ that leads to $r$ in $\D$. It also leads to some state $e_1$ almost equivalent to $r$ in $\E$.
Next, we continue our word with $\alpha_{e_1}$, witnessing a difference.
From the resulting state in $\D$, we continue with a non-empty sequence $\beta_1 \in \Sigma^+$ that brings us back to $r$. (We stay in the same SCC by construction.)
Meanwhile, we have reached some state $e_2$ almost equivalent to $r$ in $\E$.
Next, we continue our word with $\alpha_{e_2}$, witnessing a difference, and continue with a non-empty sequence $\beta_2 \in \Sigma^+$ that brings us back to $r$ in $\D$, and so forth.
We thus create an infinite sequence $\beta_0\alpha_{e_1}\beta_1\alpha_{e_2}\beta_2\alpha_{e_3}\ldots$ with infinitely many differences, which contradicts the almost equivalence of \linebreak $\D$ and $\E$.
\qed
\end{proof}

\begin{corollary}
We can construct a minimal almost equivalent automaton to a given DFA $\A$ in time quadratic in the size of $\A$.
\qed
\end{corollary}

Note that the quadratic cost occurs only for constructing the quotients of almost equivalent states.
Hence, there is a clear critical path, and improvement on this path would lead to an improvement of the overall algorithm.

It is interesting to observe that the minimal automaton almost equivalent to a weak automaton (when read as a DFA) obtained by Construction \ref{algorithm} is weak, and a language equivalent weak automata is almost equivalent.
(An automaton is called weak if it recognises the same language when read as a DBA or as a DCA, or, similarly, if all states in the same SCC have the same priority.)

\begin{theorem}
\label{theo:wmin}
The algorithm from Construction \ref{algorithm} can be used to minimise weak automata.
\end{theorem}

\begin{proof}
It suffices to show that two language equivalent (when read as DBAs or DCAs) weak automaton $\A$ and $\B$ are almost equivalent and that every almost equivalent automaton to $\A$ is weak.

As $\A$ and $\B$ are weak, the runs of an arbitrary input word $\alpha\in \Sigma^\omega$ will eventually always reside in the same SCC.
As $\A$ and $\B$ are language equivalent, all states in this SCC are either accepting (both for the respective SCC of $\A$ and of $\B$) or rejecting.
Hence, the finality of the states in the run may only differ on a finite prefix.

Let $\C$ be an automaton almost equivalent to a weak automaton $\A$.
Assuming that $\C$ is not weak, it has a (reachable) SCC that contains accepting and rejecting states.
It is simple to exploit this for constructing an input word $\alpha$ and a run $\rho$ of $\C$ on $\alpha$ such that $\rho$ contains infinitely many accepting and infinitely many rejecting states.
As $\A$ is weak, a run $\rho'$ of $\A$ on $\alpha$ will eventually always reside in the same SCC of $\A$, which contains only accepting or only rejecting states.
Hence, only finitely many states in $\rho'$ are accepting or only finitely many states in $\rho'$ are rejecting, which contradicts the assumption of almost equivalence to $\rho$.
\qed
\end{proof}

Almost equivalence can hence be read as a generalisation of language equivalence of weak automata.

\paragraph{Space Reduction for DBAs and DCAs. }
The techniques introduced for finding minimal almost equivalent automata can easily be adjusted to stronger state-space reductions for DBAs and DCAs:
If we use language equivalence for the respective automata instead of almost equivalence, the resulting automaton remains language equivalent.

\begin{theorem}
\label{theo:lmin}
Swapping quotients of almost equivalent states for the coarser quotients of language equivalent states for DBAs and DCAs in Construction \ref{algorithm} provides a language equivalent automaton $\D$, and the cost remains quadratic in the size of $\mathcal A$.
\end{theorem}

\begin{proof}
First, $\A$ and $\B$ as well as $\C$ and $\D$ are language equivalent as finite automata, and hence as DBAs and DCAs (cf.\ Lemma \ref{lem:coarser}).

To compare the language of $\B$ and $\C$, we argue along the same line as in Lemma \ref{lem:almost}:
If two states $p$ and $q$ are language equivalent, so are  $\delta(p,\sigma)$ and $\delta(q,\sigma)$ for all $\sigma$ in $\Sigma$, which implies that, for runs $r_0^b r_1^b r_2^b r_3^b \ldots$ and $r_0^c r_1^c r_2^c r_3^c \ldots$ of $\B$ and $\C$ on some $\omega$-word $\alpha$, $r_i^b$ and $r_i^c$ are language equivalent for all $i \in \omega$.
The claim then follows again from the fact that, for every run of $\C$, the states cannot go down in the pre-order $(Q_+,\succeq)$, and go up every time $\delta'(r_i^c,\sigma_i)\neq\delta(r_i^c,\sigma_i)$ holds.

The complexity remains quadratic: To determine for a DBA the ordered pairs of states for which $\scL(\B_p) \smallsetminus \scL(\B_q)$ is non-empty, we can simply construct a DPA $\scP$ with states and transitions like $\B\intercal\B$, and a priority function that maps pairs $(a,b)$ to priority $3$ if $b$ is accepting, to $2$ if $a$ is accepting while $b$ is not, and to $1$ otherwise.
It now suffices to construct $\scP$, delete the states with priority $3$, determine the non-trivial SCCs, find states in the non-trivial SCCs with priority $2$, and then consider from which states of $\scP$ they are reachable.
(This is but the deterministic version of the construction from Theorem \ref{theo:nlc}.)
Two states $p$ and $q$ are obviously language equivalent if, and only if, $\scL(\B_p) \smallsetminus \scL(\B_q)$ and $\scL(\B_p) \smallsetminus \scL(\B_q)$ are empty.

A similar construction can be devised for DCAs.
\qed
\end{proof}

An interesting corollary from the proofs of Theorems \ref{theo:lmin} and \ref{theo:rmin} is:

\begin{corollary}
\label{cor:SCC}
Minimisation techniques for DBAs or DCAs can treat the individual SCCs of the resulting automaton $\D$ individually.
\qed
\end{corollary}

An interesting aspect of this minimisation is that we can treat a local version of weak automata:
We call an SCC weak if all infinite paths within this SCC are accepting or all infinite paths within this SCC are rejecting.
For weak SCCs, we can obviously make all states accepting or rejecting, respectively, without changing the language of a DBA or DCA.

Doing so in the automaton $\C$ in from Construction \ref{algorithm} leads to all states equivalent by the respective equivalence relation (almost equivalence or language equivalence as DBA or DCA) becoming language equivalent when the automaton is read as a DFA, and are therefore merged in $\D$.
Thus, there is exactly one of these states in $\D$, and the $\D$ is locally optimal.

A further tractable minimisation would be to greedily merge states:
For an automaton $\A$ we denote with $\A^{p\triangleright q}$ the automaton that results from changing the transition function $\delta$ to $\delta'$ such that $\delta'(r,\sigma) = q$ if $\delta(r,\sigma) = p$ and $\delta'(r,\sigma)=\delta(r,\sigma)$ otherwise, choosing $q$ as initial state if $q$ was the former initial state, and removing $p$ from the state-space.
A natural tractable minimisation would be to greedily consider $\A^{p \triangleright q}$ for language equivalent states $p$ and $q$ until no further states can be merged.
Note that, by Corollary \ref{cor:SCC}, it suffices to look at the respective SCCs only, which may speed up the computation significantly.

This is even more important for exhaustive search for minimal automata, such as the SAT based methods suggested by Ehlers~\cite{Ehlers/10/SAT}.

\section{Discussion}
This report has two main results:
First, it establishes that minimising deterministic \buchi, \cobuchi\, and parity automata are NP-complete problems.

A second central contribution is the introduction of relative minimisation of DFAs, a powerful technique to minimise deterministic finite automata when allowing for minor differences in their language.
This natural minimisation problem on DFAs is strictly between the problem of hyper-minimising DFAs and minimising DBAs or DCAs and can be viewed as a generalisation of the minimisation problem of weak automata.
We show that the relative minimisation of DFAs is tractable and provide a simple quadratic algorithm.

Finally, we strengthened this algorithm by relaxing the requirement for merging states from almost to language equivalent states, which provides a promising technique to reduce the state-space of DBAs and DCAs.
This technique does not only have the potential to reduce the  state-space of the automaton significantly, it also suffices to focus on its SCCs when seeking to reduce the state-space of the automaton further.
This can be used to accelerate further reduction heuristics---like the greedy merge discussed---and exhaustive search methods alike.

\end{document}